\pdfoutput=1
\documentclass{IEEEtaes}
\usepackage[]{babel}
\usepackage{amssymb}
\usepackage{mathtools}
\usepackage{amsthm}
\usepackage{algpseudocodex} 
\usepackage{algorithm}
\usepackage{siunitx}
\usepackage{balance}
\usepackage{etoolbox}

\makeatletter
\ifboolexpr{   test {\@ifclassloaded{IEEEtaes}}}
{%
\jvol{XX}
\jnum{XX}
\jmonth{XXXXX}
\paper{1234567}
\pubyear{2022}
\doiinfo{TAES.2020.Doi Number}
\newcommand\OBR{\\}
\newcommand\OCR[1]{&#1}
}
{%
\usepackage[hidelinks]{hyperref}
\newcommand\OBR{}
\newcommand\OCR[1]{}
}
\makeatother

\newtheorem{lemma}{Lemma}
\graphicspath{{figures/}}

\title{Cache Placement in an NDN Based LEO Satellite Network Constellation}

\makeatletter
\ifboolexpr{ test {\@ifclassloaded{IEEEtaes}}}
{    \author{Rodríguez Pérez, Miguel}
    \member{Senior, IEEE}
    \affil{atlanTTic research center, Universidade de Vigo, 36212 Vigo, Spain}

    \author{Herrería Alonso, Sergio}
    \member{IEEE}
    \affil{atlanTTic research center, Universidade de Vigo, 36212 Vigo, Spain}

    \author{Suárez González, Andrés}
    \affil{atlanTTic research center, Universidade de Vigo, 36212 Vigo, Spain}

    \author{López Ardao, José Carlos}
    \affil{atlanTTic research center, Universidade de Vigo, 36212 Vigo, Spain}

    \author{Rodríguez Rubio, Raúl}
    \affil{atlanTTic research center, Universidade de Vigo, 36212 Vigo, Spain}

    \receiveddate{Manuscript received XXXXX 00, 0000; revised XXXXX 00, 0000; accepted XXXXX 00, 0000.\\
        This work has received financial support from grant
        PID2020-113240RB-I00, 
        financed by MCIN/ AEI/10.13039/501100011033, and by the Xunta de Galicia
        (Centro singular de investigación de  Galicia accreditation 2019--2022) and the
        European Union (European Regional Development Fund---ERDF).}

    \corresp{{\itshape{} (Corresponding author: M. Rodríguez-Pérez)}. All
        the authors contributed equally to this work.}

    \authoraddress{All authors are with the atlanTTic research center,
        Universidade de Vigo, 36212 Vigo, Spain
        (e-mails: \{\href{mailto:miguel@det.uvigo.gal}{miguel},
        \href{mailto:sha@det.uvigo.gal}{sha},
        \href{mailto:asuarez@det.uvigo.gal}{asuarez},
        \href{mailto:jardao@det.uvigo.gal}{jardao},
        \href{mailto:rrubio@det.uvigo.gal}{rrubio}\}@det.uvigo.gal).}

    \markboth{RODRÍGUEZ PEREZ ET AL.}{Cache Placement in an NDN Constellation}
}
{
    \author{Miguel Rodríguez-Pérez,~\IEEEmembership{Senior Member,~IEEE},\\%
        Sergio Herrería-Alonso,~\IEEEmembership{Member,~IEEE}, %
        Andrés Suárez-Gonzalez, \\José Carlos López-Ardao and Raúl
        Rodríguez-Rubio%
        \thanks{Manuscript received XXXXX 00, 0000; revised XXXXX 00, 0000;
            accepted XXXXX 00, 0000.}%
        \thanks{This work has received financial support from grant
            PID2020-113240RB-I00, 
            financed by MCIN/AEI/10.13039/501100011033, and by the Xunta de Galicia
            (Centro singular de investigación de  Galicia accreditation 2019--2022) and the
            European Union (European Regional Development Fund---ERDF).}%
        \thanks{Authors are with the atlanTTic research center,
            Universidade de Vigo, 36310 Vigo.}
    }
}
\makeatother

\DeclareMathOperator{\D}{D}
\DeclareMathOperator{\N}{N}
\DeclarePairedDelimiter\C{\lVert}{\rVert}

\newcommand\copyrighttext{%
  \footnotesize \textcopyright{} 2022 IEEE\@. Personal use of this material is
  permitted.  Permission from IEEE must be obtained for all other uses, in any
  current or future media, including reprinting/republishing this material for
  advertising or promotional purposes, creating new collective works, for resale
  or redistribution to servers or lists, or reuse of any copyrighted component
  of this work in other works. DOI:\ \href{https://doi.org/10.1109/TAES.2022.3227530}{10.1109/TAES.2022.3227530}
  } \newcommand\copyrightnotice{%
\begin{tikzpicture}[remember picture,overlay]
\node[anchor=south,,xshift=60pt,yshift=68pt] at (current page.south) {\fbox{\parbox{\dimexpr\textwidth-\fboxsep-\fboxrule\relax}{\copyrighttext}}};
\end{tikzpicture}%
}

\begin{document}

\maketitle
\copyrightnotice{}

\begin{abstract}
    The efforts to replace the successful, albeit aging, TCP/IP Internet
    architecture with a better suited one have driving research interest to
    information-centric alternatives. The Named Data Networking (NDN)
    architecture is probably one of the main contenders to become the network
    layer of the future Internet thanks to its inbuilt support for mobility,
    in-network caching, security and, in general, for being better adapted to
    the needs of current network applications. At the same time, massive
    satellite constellations are currently being deployed in low Earth orbits
    (LEO) to provide a backend for network connectivity. It is expected that,
    very soon, these constellations will function as proper networks thanks to
    inter-satellite communication links. These new satellite networks will be
    able to benefit from their greenfield status and the new network
    architectures. In this paper we analyze how to deploy the network caches of
    an NDN-based LEO satellite network. In particular, we show how we can
    jointly select the most appropriate caching nodes for each piece of content
    and how to forward data across the constellation in two simple alternative
    ways. Performance results show that the caching and forwarding strategies
    proposed reduce path lengths up to a third with just a few caching nodes
    while, simultaneously, helping to spread the load along the network.
\end{abstract}


\section{Introduction}%
\label{sec:introduction}

\IEEEPARstart{D}{uring} the last few years we are witnessing a commercial race
for providing low-latency, high-bandwidth Internet access with the help of
massive constellations of satellites in low Earth orbit (LEO). Probably, the
best well-known example is SpaceX's Startlink network, operating, as of July
2022, about \num{2500} satellites~\cite{mcdowell_jonathan_startlink_2022}.
However, there are many other competing networks in different completion
phases~\cite{uk_space_agency_18m_2019,hindin_application_2019,telesat_telesat_2020}.
Although the satellites in these networks will initially act just as packet
relays between pairs of ground stations, the main benefits will be obtained
when traffic can travel directly across the constellation with the help of
inter-satellite links (ISL)~\cite{bhattacherjee_network_2019,pan_scalable_2021}.

Almost simultaneously to this commercial interest in satellite networks, the
network research community has started paying attention to a new networking
paradigm focused not on providing connectivity between distant devices, but on
the data acquisition itself. Several proposals, under the umbrella of the
Information Centric Networking (ICN) paradigm~\cite{koponen_data-oriented_2007},
try to create a new global network architecture that one day may replace the
current TCP/IP Internet. Proponents of these architectures claim that they are
better suited to current applications due to their natural support for consumer
mobility~\cite{xia_adapting_2021}, in-network caching~\cite{ghasemi_far_2020},
multicast transmission~\cite{lederer_adaptive_2014,rainer_investigating_2016}
and in-built security~\cite{tourani_security_2018}.

In this work we focus on the network caching characteristics of an ICN network
when applied to a massive LEO satellite constellation. We will assume that each
satellite will carry four inter-satellite links (ISL) to communicate with its
four closest neighbors. Although linking with the closest available satellites
is not the only feasible alternative~\cite{bhattacherjee_network_2019}, neither
the best one, it helps to keep things manageable and has already been proposed
by several works~\cite{handley_delay_2018,chaudhry_laser_2021}. As for the
actual ICN architecture, we will use the Named-Data Networking (NDN)
proposal~\cite{zhang_named_2014}, as it is already in a mature
state. In an NDN network, data is directly addressed by its unique name, rather
than by its location, as done by IP networks. For this, the network uses two
different kinds of packets: \emph{Interest} packets, that carry a request for a
named piece of data; and \emph{Data} packets, carrying the actual data back to
the requester(s) 
following the reverse path used by the Interest packet.


In this paper we explore the regular topology of satellite networks to give answer to what
appear to be two contradictory but highly desired characteristics. On the one
hand, we want to spread the load on the network so that information from
different sources follows disjoint, albeit equal-cost, network paths, thus maximizing network capacity. On the other hand, traffic
of Interest packets for the same information should converge at common nodes so
as to maximize the effectiveness of in-network caching. Thus, we provide:
\begin{enumerate}
    \item A simple Interest \emph{forwarding strategy} that, while following the
          shortest path to a target producer, always finds the nearest cache while
          spreading the load evenly across the whole network;
    \item a discussion about which nodes should be involved in caching data from
          each producer to maximize cache hit-rate;
    \item and an algorithm that finds the best possible cache locations for a
          given satellite constellation working in tandem with our forwarding
          strategy.
\end{enumerate}
To assess the adequacy of our
proposed solution, we simulate a large LEO constellation network to find that,
with just a few caching nodes, the number of transmissions required to obtain a
named piece of data is just halved. This has profound positive effects both in the used
capacity and the incurred transmission delay.

The rest of this paper is organized as follows.
Section~\ref{sec:problem-description} describes the considered scenario. Then,
Section~\ref{sec:cache-placement} discusses alternatives for cache placement.
The experimental results are shown in Section~\ref{sec:results}. Finally, we lay
out our conclusions in Section~\ref{sec:conclusions}.

\section{Problem Description}%
\label{sec:problem-description}

We consider a scenario consisting in a LEO satellite constellation with
inter-satellite communication capabilities and a set of ground stations acting
as the entry (exit) points of the orbiting network. All orbiting and terrestrial
nodes run an instance of an NDN network protocol.

As shown in
Fig.~\ref{fig:constellation-earth-render}, satellites in a usual LEO constellation are organized in various orbits sharing
the same inclination (planes). The relative positions of individual satellites
in the same orbital plane is kept relatively stable and this permits to keep
connections with both the preceding and following satellites in the same orbit.
Moreover, the distance between the orbits is also stable, permitting also
connections with the nearest satellites in the two immediate neighboring orbital
planes. Thus, assuming four inter-satellite links per satellite, this results in
a grid-like topology for the orbiting part of the network, like the one shown in
Fig.~\ref{fig:constellation-example-grid}.

\begin{figure}
    \centering
    \includegraphics[width=.9\columnwidth]{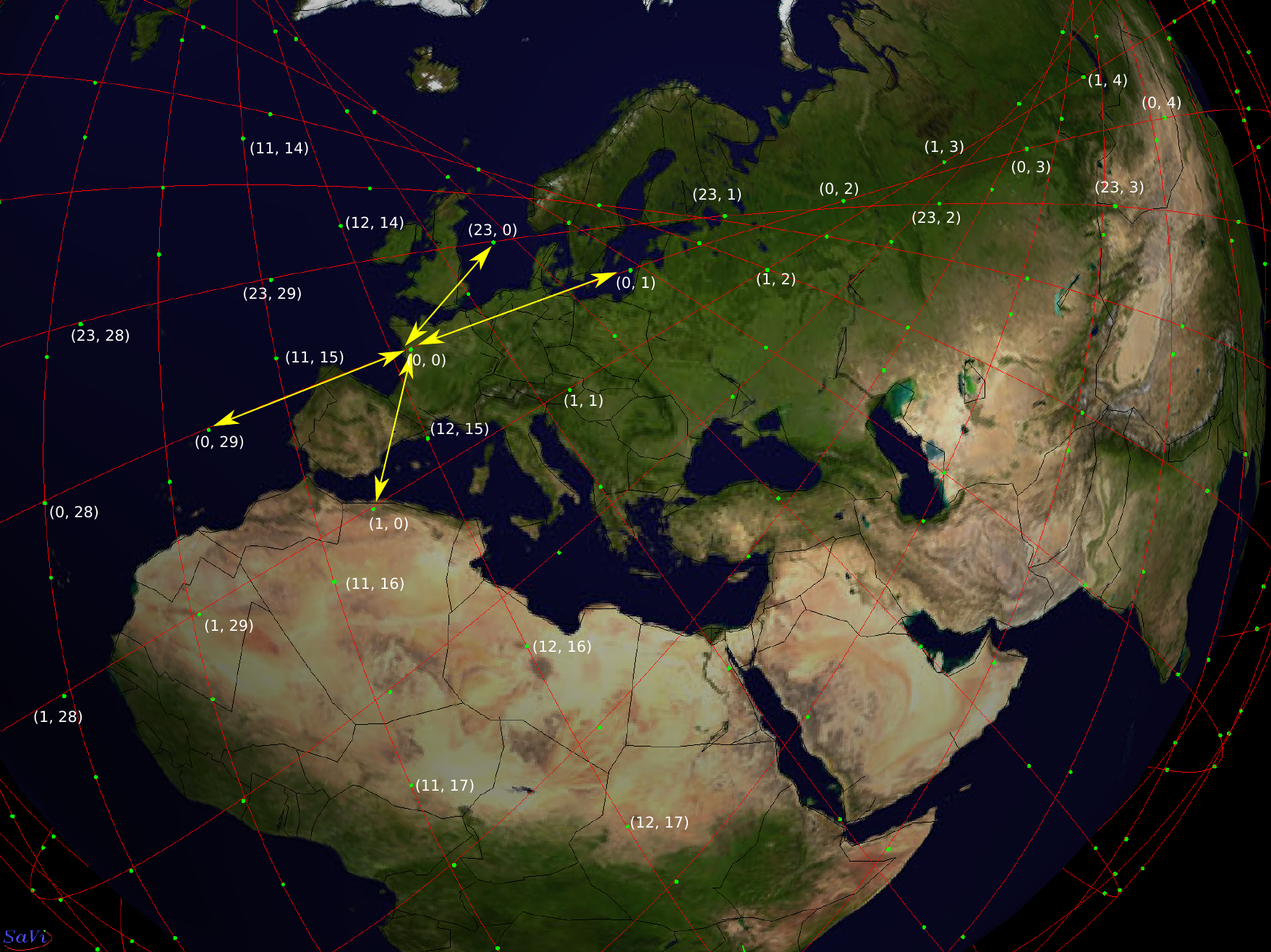}
    \caption{A satellite constellation with 24 orbital planes and 30 satellites
        in each plane with an inclination of 60º as seen from space. We have
        represented the satellite identifiers and, with yellow lines, the four
        possible ISL links of the satellite located at coordinates \((0,0)\).}%
    \label{fig:constellation-earth-render}
\end{figure}

As the connectivity between ground stations and their orbiting counterparts is
subject to frequent changes, we employ two different routing strategies. In the
terrestrial part, the satellite network topology can be simply ignored since the
LEO constellation will be used
as a backbone providing connectivity between any pair of ground stations. Thus,
for some prefixes, the LEO network will provide the best path and will be used
for routing some Interest packets. In that case, the corresponding ground
station will simply relay the Interest packet to any reachable overhead
satellite. Then, the satellites will forward the Interest packet to the nearest
most appropriate terrestrial node, as all the producers are assumed to stay on
the ground. Due to the grid-like connectivity topology of the LEO network, once
the target ground node has been located,\footnote{We do not delve into how this
    information is procured, but it can be readily obtained if nodes use a
    link-state routing protocol, for instance.} the routing is straightforward as we
shall see later.

For the rest of the article we will consider a simplified satellite network,
made up of a single constellation where each satellite has four ISL links, two
with the nearest satellites in the same plane (the immediately ones in
\emph{front} of it and \emph{behind} it), and another two with the
nearest satellites in the plane to port and starboard. Recall that the relative
positions of the satellites do not change as they travel through their orbits.
Although the connectivity pattern gets obscured at the northernmost (and
southernmost) regions of the constellation due to the increased density, the
ordering of the satellites in the same orbital plane is not changed. The same
happens with the relative positions of the successive planes. Therefore, even
though successive orbital planes have some offset, the resulting connectivity
pattern (if the metric is the number of hops) can be represented as a grid
\(\mathcal{G}\). Let then \(\mathcal{G} = \{\mathcal{N}, \mathcal{E}\} \), where
\(\mathcal{N} = \{(x,y), x \in \mathbb{Z}/n_{\mathrm p}\mathbb{Z}, y \in
\mathbb{Z}/n_{\mathrm s}\mathbb{Z} \} \) is the set of satellites, represented
by their coordinates in the constellation, \(n_{\mathrm p}\) and \(n_{\mathrm
        s}\) are, respectively, the number of orbital planes and the number of
satellites per plane, and \(\mathbb{Z}/m\mathbb{Z}\) is the ring of integers
modulo~\(m\). \(\mathcal{E} = \{(n_i, n_j), n_i, n_j \in \mathcal{N} \,|\,
d(n_i,n_j)=1\} \) is the set of ISLs inside the constellation, with \(d(\cdot)
\) a modified taxicab metric that takes into account the modular nature of the
scenario,\footnote{We need to use modular arithmetic to account for the fact
    that there is connectivity both before the first and the last satellite of a
    plane and between the first and the last planes themselves.} so
\begin{equation}%
    \label{eq:taxicab-mod}
    \begin{multlined}
        d(n_i, n_j) = d((x_i, y_i), (x_j, y_j))\OBR
        = \min \{ x_i \ominus x_j, x_j \ominus x_i \}
        + \min \{ y_i \ominus y_j, y_j \ominus y_i \},
    \end{multlined}
\end{equation}
where \(\ominus \) is the subtraction operation in \(\mathbb{Z}/m\mathbb{Z}\).

Without loss of generality, and taking into account the symmetric nature of the
satellite network, we will always consider that the exit node is the node
located at coordinates \((0,0)\) and only examine the behavior of those nodes
located in the top-right quadrant.\footnote{The exit node will be different for
    different prefixes, so there is not a central node of the network, but a central
    node with respect to a given prefix.}
Figure~\ref{fig:constellation-example-grid} depicts an example of such a grid
network, where the nodes in the top-right quadrant have been highlighted.
\begin{figure}
    \centering
    \includegraphics[width=.9\columnwidth]{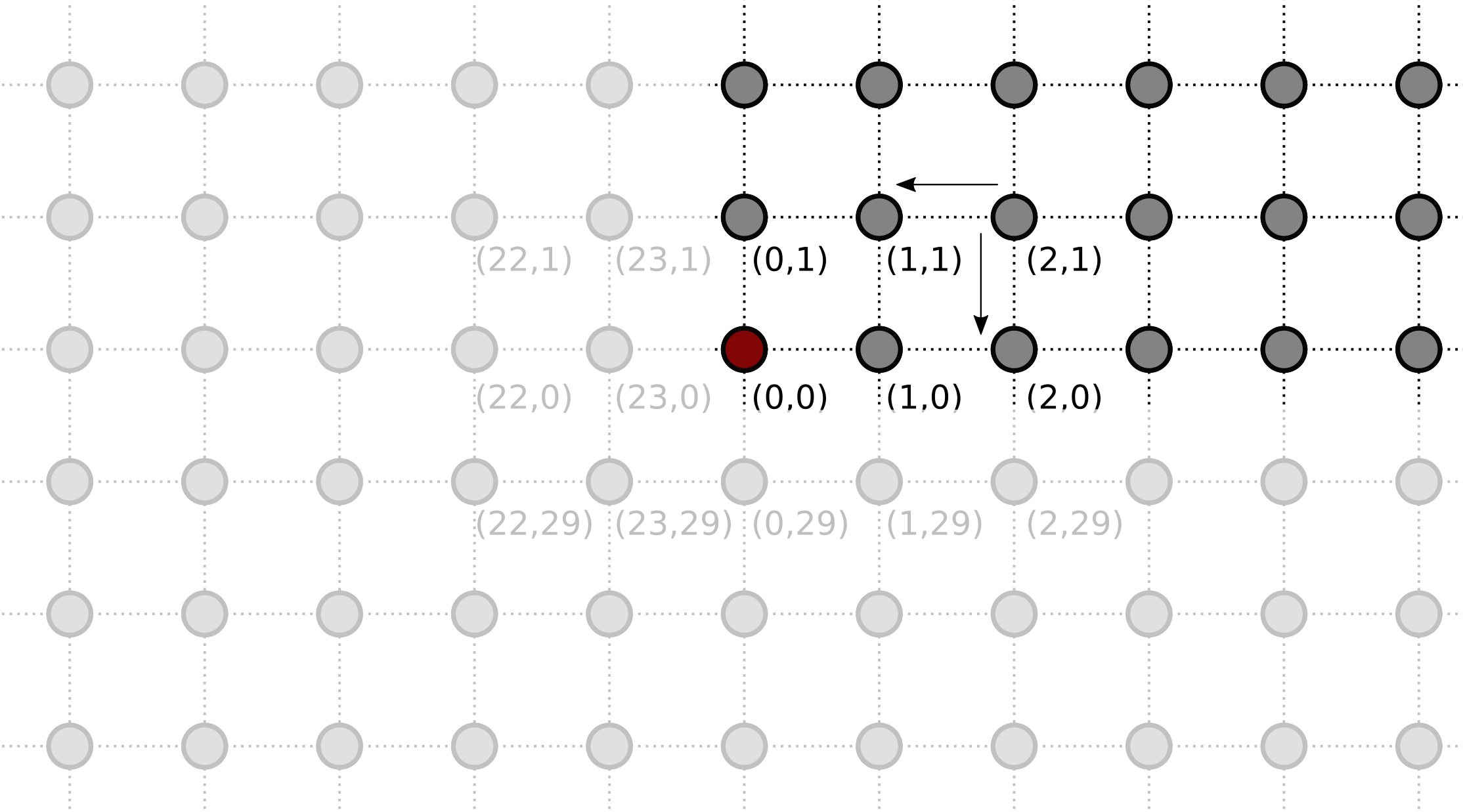}
    \caption{A grid-like representation of the links between neighboring
        satellites of the constellation of
        Fig.~\ref{fig:constellation-earth-render}.}%
    \label{fig:constellation-example-grid}
\end{figure}

\subsection{Forwarding Strategy}%
\label{sec:fwd-stragety}

Before we delve into the specifics of the forwarding strategy, it may be good to
recapitulate how the network layer of the NDN architecture works. In the NDN
architecture no data can be sent to the network unless a node has previously
asked for it. This \emph{pull}-based communication approach is in sharp contrast
with the IP architecture, where any node can \emph{push} data to the network as
long as it knows a destination address. To accommodate to this pull-based
operating manner, NDN defines two different types of network layer packets: the
\emph{Interest} and the \emph{Data} packets. Thus, when a consumer needs to
obtain some data from another node, it sends to the network an Interest packet
that specifies the \emph{name} of the requested data. Then, the NDN routers
forward this Interest packet according to the named requested content (usually
just the prefix of the name) and their configured forwarding mechanism. However,
if they have previously seen this particular named Data, they can directly
return a copy from their cache, if they happen to have stored it. If there is no
copy, they forward the Interest packet to the next NDN router and store the
information about it in a temporal table of pending interests,
called the Pending Interest Table, or PIT\@. Finally, when the Interest reaches a
\emph{producer} for the requested data, it answers with a Data packet. This Data
packet flows back to the requesting consumer(s) 
following the reverse path taken
by the Interest packet. As the Data packet reaches every intermediate NDN
router, they use the information stored in the PIT to forward it to the
destination while they optionally keep a copy stored in their local cache.

Due to the grid-like topology of the satellite network, forwarding an Interest
packet towards the exit node is just a matter of selecting any of the two closer
neighbors. That is, for a satellite in the top-right quadrant, at location
\((i,j)\)  relative to the exit gateway, that means using either \((i-1,j)\) or
\((i,j-1)\) as the next hop node, as exemplified for node \((2,1)\) in
Fig.~\ref{fig:constellation-example-grid}. From a pure forwarding perspective,
both neighbors represent an optimal choice.

However, nodes in an NDN network may cache the contents of any previous Data
packet. If the Interest packet is forwarded to a node already holding a copy of
the requested information, it does not need to be forwarded further and, instead,
the node can reply immediately with the copy. This results in shorter delays and avoids
unnecessary transmissions along the satellite network and even in the satellite
to ground exit link. So, it is important to forward Interest packets in a way
that Interest packets from different nodes converge at some common caching
nodes. However, to avoid overwhelming the storage capacity of the caches, it is
better for different nodes to cache the contents of different prefixes. At the
same time, given that all the links in the satellite network have identical
characteristics, it is important to spread traffic to maximize the
aggregated capacity. Luckily, all these conditions can be simultaneously met if
the decision of whether to cache a piece of content depends on the location of
the candidate caching node relative to the exit location (the center node for a given
prefix). Keep in mind that non-location aware cache management algorithms (like
those based on popularity, freshness\dots{}) can be used in tandem with a
location-aware one to also influence the decision to cache a piece of content.

With all these considerations we propose the following two simple rules for
deciding the next node in forwarding decisions:

\begin{enumerate}
    \item Interest packets must be forwarded without increasing distance to the
          primary source (the one at the origin).
    \item Interest packets should be forwarded to the closest \emph{allowed} cache.
\end{enumerate}

The first rule avoids looping. Even though a caching node closer to the current
node, but further away from the producer, may hold a copy of the content, moving
away from the center causes loops when the content is not found in the
cache.\footnote{If nodes closer to the producer than the presumed cache forward
    traffic to it, then traffic from this caching node will not be able to reach the
    producer when the caching node forwards an Interest for the data for the first
    time since its immediate neighbors would always forward the Interest back to it.
    Certainly, one can devise mechanisms that can solve this scenario, but we feel
    that the added complexity is not worth it.} The second rule not only makes the
traffic converge to a caching node, but also helps to spread the traffic across
the network. Recall that caching nodes are set at positions relative to the
network center for each named prefix.

As a result of these rules, different sets of nodes forward to (are served by)
different caching nodes. Figure~\ref{fig:grid-with-caches-generic} shows the
nodes of the top-right quadrant of a network with three caches and a producer at
the bottom-left corner.
\begin{figure}
    \centering
    \includegraphics[width=.9\columnwidth]{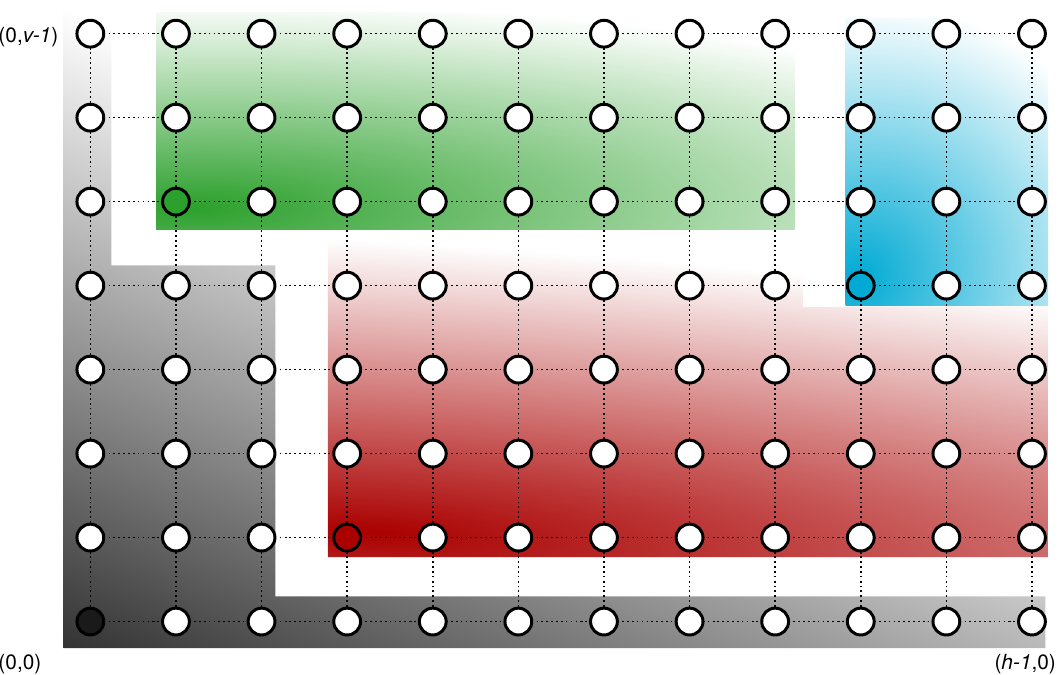}
    \caption{Top-right quarter of the constellation grid, showing the location
        of the producer at the origin and three more arbitrarily placed caches
        (solid colors). The colored regions show the serving cache for each
        node. The satellite network is composed of \(n_{\mathrm p} = 2h\)
        orbital planes with \(n_{\mathrm s} = 2v\) satellites in each plane.}%
    \label{fig:grid-with-caches-generic}
\end{figure}
All nodes in the gray area are served by the original producer at \((0,0)\), as
being served by either the red caching node at \((3,1)\) or the green one at
\((1,5)\) would entail either increasing the value of one of their coordinates.
This would mean forwarding Interests farther from the producer, violating rule 1
above. Note that nodes in the red, green and blue areas are served by nodes
\((3,1)\), \((1,5)\), and \((9,4)\), respectively, since the caching
node in their corresponding region is the closest cache (rule 2) and forwarding to
it gets closer to the producer (rule 1).

What follows is a discussion about where to place caching nodes for each prefix.

\section{Cache Placement}%
\label{sec:cache-placement}

Even though every node is free to opportunistically store any Data packet, we
must select the nodes \emph{responsible} for caching the contents of a given
producer. Their location should become a convergence point of disjoint paths
towards the producer. Moreover, their locations should also be such that they
minimize the average number of hops that an Interest packet must be forwarded
before it encounters such a caching node. This metric clearly saves transmission
capacity and minimizes delay.

\subsection{Regular Cache Placement}%
\label{sec:sq-cache-placement}

One natural way to place the caches arises from subdividing the original region,
i.e., the set of nodes served by the producer, into \(r^2\) identical
subregions, like in Fig.~\ref{fig:sub-rectangular-example}.
\begin{figure}
    \centering
    \includegraphics[width=.9\columnwidth]{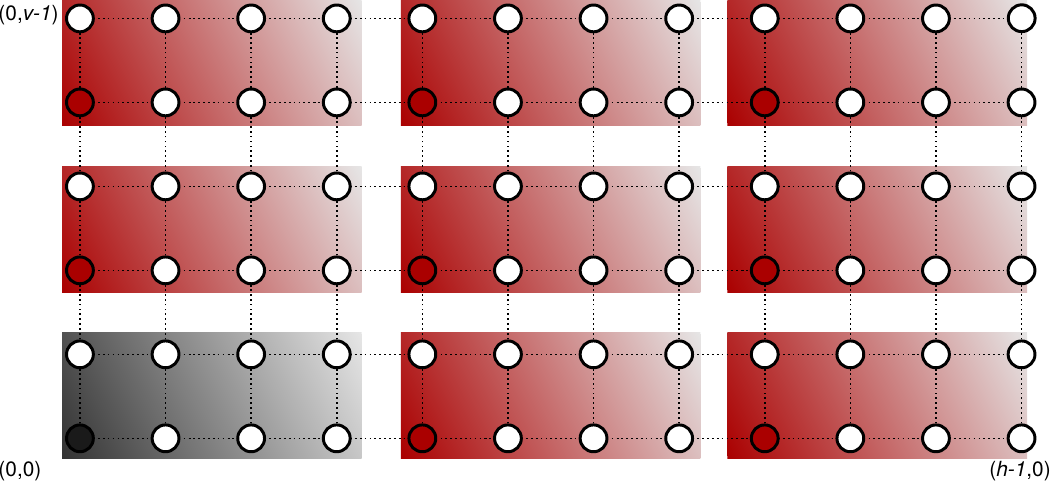}
    \caption{Top-right quarter of the constellation grid showing the location of
        caches as a result of subdividing the original region into \(r^2=9\)
        smaller identical subregions.}%
    \label{fig:sub-rectangular-example}
\end{figure}
It is straightforward to check that in such a placement, each satellite is at
most \((h+v)/r-2\) hops away from the nearest copy, where \(v=n_{\mathrm s}/2\)
is the height of the resulting top-right grid and \(h=n_{\mathrm p}/2\) its
width. It is also quite easy to get the average distance to a cache in the
network, or equivalently, in a single subregion, as
\begin{equation}
    \label{eq:av-distance-sq}
    \D^{\mathrm{reg}} = \frac{\sum_{i=0}^{\frac{h}{r}-1}\sum_{j=0}^{\frac{v}{r}-1}(i+j)}{\frac{h}{r}\frac{v}{r}}
    = \frac{h+v}{2r}-1,
\end{equation}
where we assume that both \(h\) and \(v\) are integer multiples of \(r\) to keep
the expression simple. As all subregions are
identical,~\eqref{eq:av-distance-sq} is approximately the average distance in the whole
network.\footnote{Non caching nodes in the axes are taken into account twice.
    However, the global effect in the average distance is very small for
    sufficiently large networks.}

We have also to consider how many total caches are needed for such an
arrangement in the full satellite topology. It can be easily deduced that, for
\(r^2\) subregions in a single quadrant, we need \((r-1)\) caching nodes in
each axe and \({(r-1)}^2\) caching nodes outside the axes. As there are four
such quadrants, each pair sharing a semiaxis, the total number of caching nodes
in the topology is
\begin{equation}
    \label{eq:total-caches-sq}
    \N^{\mathrm{reg}} = 4{(r-1)}^2+4(r-1)=4r(r-1).
\end{equation}

\subsection{In-Axes Cache Placement}%
\label{sec:axes-cache-placement}

There is also a natural cache placement strategy that consists in considering
only nodes in the axes. In this way, routing becomes even more straightforward.
Now, Interest packets can be forwarded directly first to the nearest axis, and
then to the producer. This ensures that the packet will come across a caching
node in the process. Obviously, if the forwarding node is aware of the precise
location of the nearest cache, or of on which axis it resides, it can still use
this information to forward the Interest packet even to the furthest axis if that
results in reaching a closer caching node.

\begin{figure}
    \centering
    \includegraphics[width=.9\columnwidth]{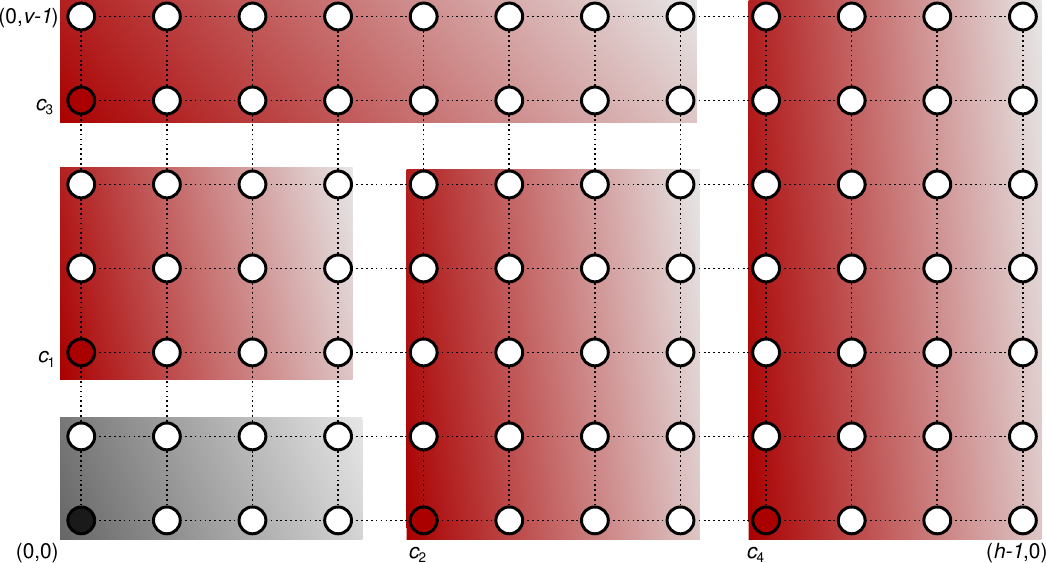}
    \caption{Top-right quarter of the constellation grid showing the regions
        covered by each of the caches located in the axes. In this specific
        scenario \(c_1 = v_1\) and \(c_2=h_1\).}%
    \label{fig:in-axes-example}
\end{figure}
Figure~\ref{fig:in-axes-example} shows the area of influence of several caching
nodes located in the axes of a grid network. Note how the further away from the
center (where the producer is located), the greater area of coverage of each
cache. The location of the caching nodes must be carefully chosen, to
minimize the average distance to a cache in the network.

We can formalize the area \emph{covered} by each cache in the following way.
Consider two sets of ordered caches, one in the vertical axe \(\mathcal V =
\{(0,v_1),\dots,(0,v_V)\} \) and the corresponding set in the horizontal one
\(\mathcal H = \{(h_1,0),\dots,(h_H,0)\} \). Then, the set of nodes served by
node \((h_i, 0) \in \mathcal H\)---resp. \((0, v_i) \in \mathcal V\)---is
\(\mathcal A(h_i) = \{(x,y) \,|\, h_i \le x < h_{i+1},\  y < v_k\} \), where
\(v_k = \max(v_j | (0, v_j) \in \mathcal V \cup \{(0, v)\} \text{ and  } v_j < h_{i+1}) \)
and \(h_{H+1} = (h,0) \)---resp.\ for \(\mathcal A(v_i) \)---. To obtain the
average distance to a cache, we have to first obtain the sum of the distances
from every node to its nearest cache and divide it by the number of nodes. As we
can see, the regions are defined by three parameters: the positions of two
consecutive caches in one axis, and a single cache in the other axis. If we define
\(\C{\mathcal A(h_i)}=\C{(h_i,h_{i+1},v_k)}\) (resp. \(\C{\mathcal
    A(v_i)}=\C{(v_i,v_{i+1},h_k)}\)) as the sum of the distances in the \(\mathcal
A(h_i)\) region, we get
\begin{equation}%
    \label{eq:cost-region-axes}
    \begin{split}
        \C{(a_1,a_2,b)} & =
        \sum_{i=a_1}^{a_2-1}\sum_{j=0}^{b-1}d((i,j),(a_1,0)) \OBR
        \OCR{} =  \sum_{i=a_1}^{a_2-1}\sum_{j=0}^{b-1}(i+j-a_1) \OBR
        \OCR{} =  b(a_2-a_1)\frac{a_2-a_1+b-2}{2}.
    \end{split}
\end{equation}
The average distance in this setup is thus
\begin{equation}%
    \label{eq:avdistance-region-axes-total}
    \begin{split}
        \D^{\mathrm{axes}}& = \frac{\C{\mathcal A(0)} +  \C{\mathcal H} + \C{\mathcal V}}{h v}\OBR
        \OCR{} =\frac{\C{\mathcal A(0)}}{h v}
        + \frac{\sum_{h_i\in\mathcal H}\C{\mathcal A(h_i)} + \sum_{v_i\in\mathcal V}\C{\mathcal A(v_i)}}{h v},
    \end{split}
\end{equation}
where \(\mathcal A(0)\) is the set of nodes served directly by the producer.\footnote{As in the case of the regular cache placement, the value is exact for a single quadrant. For the whole network it is a good approximation when the constellation is large enough.}

We prove in Appendix~\ref{sec:appendix} that placing the caches in an
interleaved way (\(\ldots < h_i < v_j < h_{i+1} < v_{j+1}<\ldots \))
minimizes~\eqref{eq:avdistance-region-axes-total}.

\subsection{A Fast Algorithm to Compute the Optimal in-Axes Cache Locations}

Our next step is to obtain the optimal location of the caches in the axes for a
given number \(N\) of total caches. To keep the notation simple, we will ignore
the 0-valued dimension in each of the caches, so that \((h_i,0)\) becomes
directly \(h_i\)---resp. \((0,v_i)\) becomes \(v_i\)---and, using the fact that
the caches are interleaved, work directly with the vector \(\mathcal C = \{c_1,
c_2, \dots, c_N\} = \mathcal H \cup \mathcal V \), where each element represents
either a cache location in \(\mathcal H\) or in \(\mathcal V\). The problem is
thus to find \(\mathcal C\) that
minimizes~\eqref{eq:avdistance-region-axes-total}, that is, that minimizes
\begin{equation}%
    \label{eq:cost-caches-axes-expanded}
    \begin{split}
        \C{\mathcal C}& = \C{(0, c_2, c_1)} + \C{(c_1,c_3,c_2)}\OBR
        \OCR{\quad} + \C{(c_2,c_4,c_3)} + \C{(c_3,c_5,c_4)} + \dots.
    \end{split}
\end{equation}
Sadly, it is not possible to minimize this function analytically, but the
procedure shown next finds the optimal locations in less than \(h N\)
iterations.\footnote{From now onwards, we will assume for clarity a square
    network with \(h=v\), although the results can be easily extended to the general
    scenario.} The procedure detailed in Algorithm~\ref{alg:cache-placement-axes}
sets the initial location of the caches to those closest to the producer. That
is, for a solution with \(N\) caching nodes, they initially hold positions
\(c_i=i, 1 \le i \le N\). After the initialization, our algorithm displaces the
furthest cache to more distant locations until the cost stops decreasing (lines
6--14). Then, it tries with the next one (lines 5--14), and so on. When all the
caches have been tried, it tries to move the caches again starting with the
furthest one (loop of lines 3--15) if the inner loops have found
a better solution (\emph{finish} set to \emph{false} in line 12). If
\emph{finish} is still \emph{true} after the loop in lines 5--14 ends, then no
better solution has been found and the procedure ends.
\begin{algorithm}
    \begin{algorithmic}[1]
        \State{} \(\mathcal C = \{1, 2, 3, \dots, N\}, c_{N+1} = h \)
        \State{} \(lowest\_cost \gets \C{\mathcal C}\)

        \Repeat{}
        \State{} \(finish \gets \mathrm{true}\)
        \ForAll{\(i \in \{N, \dots, 1\} \)}
        \While{\(c_i + 1 < c_{i+1}\)}
        \State{} \(\mathcal C' = \{c_1, c_2, \dots, c_i + 1, c_{i+1}, \dots,
        c_N\} \)
        \State{} \(current\_cost \gets \C{\mathcal C'}\)
        \If{\(current\_cost < lowest\_cost\)}
        \State{} \(lowest\_cost \gets current\_cost\)
        \State{} \(\mathcal C \gets \mathcal C'\)
        \State{} \(finish \gets \mathrm{false}\)
        \Else{}
        \State{} \textbf{break} while loop
        \EndIf{}
        \EndWhile{}
        \EndFor{}
        \Until{\(finish = \mathrm{true}\)}
    \end{algorithmic}
    \caption{Algorithm for finding the best cache locations along the axes.}%
    \label{alg:cache-placement-axes}
\end{algorithm}

This simple procedure finds the optimal cache locations as can be derived from
Lemma~\ref{lm:proof-optimal} that shows that if moving a cache further from the
producer reduces the total cost and moving the previous one also reduces the
cost, then moving both reduces the cost even more. This ensures that when we
advance an \emph{outer} cache to find a local optimum we are not going to miss a
global optimum  resulting from moving only \emph{inner} caches.
\begin{lemma}%
    \label{lm:proof-optimal}
    Let \(\mathcal C=\{c_1, c_2, \dots, c_i, \dots, c_n\} \) be a set of ordered
    cache locations and \(\Delta_i \mathcal C = \{c_1, c_2, \dots, c_i + 1,
    \dots, c_n\} \). Then, if both \(\C{\Delta_i \mathcal C} < \C{\mathcal
        C}\) and \(\C{\Delta_{i-1}\mathcal C} < \C{\mathcal C}\), it holds
    that \(\C{\Delta_{i-1}\Delta_i\mathcal C} < \C{\Delta_{i-1}\mathcal
        C} \).
\end{lemma}
\begin{proof}
    For \(\C{\Delta_{i-1}\Delta_i\mathcal C} < \C{\Delta_{i-1}(\mathcal
        C)}\) to be true, it must hold that \(\C{\Delta_{i-1}\mathcal C} -
    \C{\Delta_{i-1}\Delta_i\mathcal C}>0\).

    According to~\eqref{eq:cost-caches-axes-expanded},
    \begin{gather}%
        \label{eq:proof-c_p_c_astk_expanded}
        \begin{split}
            \C{\Delta_{i-1}\mathcal C}& = \C{(0,c_2,c_1)} + \C{(c_1,c_3,c_2)}+\C{(c_2,c_4,c_3)}\OBR
            \OCR{\quad} +\cdots + \C{(c_{i-3},c_{i-1}+1,c_{i-2})}\\
            & \quad + \C{(c_{i-2},c_{i},c_{i-1}+1)}\OBR
            \OCR{\quad} + \C{(c_{i-1}+1,c_{i+1},c_{i})}\OBR
            \OCR{\quad} + \C{(c_{i},c_{i+2},c_{i+1})}+\cdots \\
        \end{split}\\
        \begin{split}
            \C{\Delta_{i-1}\Delta_i\mathcal C} &= \C{(0,c_2,c_1)} + \C{(c_1,c_3,c_2)}+\C{(c_2,c_4,c_3)}\OBR
            \OCR{\quad} +\cdots + \C{(c_{i-3},c_{i-1}+1,c_{i-2})} \\
            & \quad + \C{(c_{i-2},c_{i}+1,c_{i-1}+1)}\OBR
            \OCR{\quad} + \C{(c_{i-1}+1,c_{i+1},c_{i}+1)}\OBR
            \OCR{\quad} + \C{(c_{i}+1,c_{i+2},c_{i+1})}+\cdots.
        \end{split}
    \end{gather}

    So \(\C{\Delta_{i-1}\mathcal C} - \C{\Delta_{i-1}\Delta_i\mathcal C}>0\) expands to

    \begin{multline}
        \label{eq:proof-c_p_c_astk_expanded_full}
        \C{\Delta_{i-1}\mathcal C} - \C{\Delta_{i-1}\Delta_i\mathcal C} =
        \C{(c_{i-2},c_{i},c_{i-1}+1)}\OBR
        - \C{(c_{i-2},c_{i}+1,c_{i-1}+1)}
        + \C{(c_{i-1}+1,c_{i+1},c_{i})}\\
        - \C{(c_{i-1}+1,c_{i+1},c_{i}+1)}
        + \C{(c_{i},c_{i+2},c_{i+1})}\OBR
        - \C{(c_{i}+1,c_{i+2},c_{i+1})}.
    \end{multline}
    If we apply~\eqref{eq:cost-region-axes} and perform some straightforward
    simplifications, we get that \(\C{\Delta_{i-1}\mathcal C} -
    \C{\Delta_{i-1}\Delta_i\mathcal C}>0\) iif
    \begin{gather}
        \label{eq:proof-caches-condition-expanded}
        \begin{split}
            0& <  c_{i-1}(c_{i+1}+c_{i-2}-c_{i-1}-2)\OBR
            \OCR{\quad} + c_{i+1}(c_{i+2}-2c_i)+c_{i-2}-1
        \end{split}\\
        \begin{split}
            2c_{i-1}-c_{i-2}+1& < c_{i-1}(c_{i+1}+c_{i-2}-c_{i-1})\OBR
            \OCR{\quad} + c_{i+1}(c_{i+2}-2c_i).
        \end{split}
    \end{gather}

    We also know that \(\C{\Delta_i\mathcal C} < \C{\mathcal C}\). If we
    repeat the same procedure as before, we get that
    \begin{equation}
        \label{eq:proof-caches-cn_better-expanded}
        \begin{split}
            0& < c_{i-1}(c_{i+1}+c_{i-2}-c_{i-1})+c_{i+1}(c_{i+2}-2c_i-1) \\
            c_{i+1}& < c_{i-1}(c_{i+1}+c_{i-2}-c_{i-1})+c_{i+1}(c_{i+2}-2c_i).
        \end{split}
    \end{equation}
    So,~\eqref{eq:proof-caches-condition-expanded} is true iff
    \begin{equation}
        \label{eq:proof-caches-condition-after_cn}
        c_{i+1} \ge 2c_{i-1}-c_{i-2}+1.
    \end{equation}
    Considering that \(\C{\Delta_{i-1}\mathcal C} < \C{\mathcal C}\), we find
    that
    \begin{equation}
        \label{eq:proof-caches-cp_better_expanded}
        \begin{split}
            0& < c_{i-2}(c_i+c_{i-3}-c_{i-2})+c_i(c_{i+1}-2c_{i-1}-1) \\
            c_i c_{i+1}& > c_i(2c_{i-1}-c_{i-2}+1)+c_{i-2}(c_{i-2}-c_{i-3}).
        \end{split}
    \end{equation}
    So, from~\eqref{eq:proof-caches-cp_better_expanded} we obtain that
    \begin{equation}
        \label{eq:proof-caches-final-step}
        \begin{split}
            c_{i+1}& > 2c_{i-1}-c_{i-2}+1+\frac{c_{i-2}(c_{i-2}-c_{i-3})}{c_i}  \\
            & > 2c_{i-1}-c_{i-2}+1,
        \end{split}
    \end{equation}
    as \(c_{i-2} > c_{i-3}\).
\end{proof}






It is easy to extend this result for arbitrary advancements of either
\(c_{i-1}\) or \(c_i\). On the one hand, \(\C{\Delta_{i-1}^j\Delta_i \mathcal C}
= \C{\overbrace{\Delta_{i-1}\Delta_{i-1}\dots\Delta_{i-1}}^j\Delta_i \mathcal C}
< \C{\Delta_{i-1}^j{\mathcal C}} \) by Lemma~\ref{lm:proof-optimal}, as long as
\(c_{i-1} + j < c_i\), because Lemma~\ref{lm:proof-optimal} does not place any other restriction on the value of
\(c_{i-1}\). On the other
hand, \(\C{\Delta_{i-1}\Delta_i^k\mathcal C} <
\C{\Delta_{i-1}\Delta_i^{k-1}\mathcal C}\) if, as by hypothesis,
\(\C{\Delta_i^k\mathcal C} < \C{\Delta_i^{k-1}\mathcal C}\).

\section{Experimental Results}%
\label{sec:results}

We have tested the previous results with the help of a newly developed routing
module~\cite{rodriguezperez_ndngridcache_2022} for the
ndnSIM~\cite{noauthor_ndnsim_2022} NDN network simulator. We have also released
the software used to calculate the optimal cache
placement~\cite{rodriguezperez_fastgridcache_2022}.

We have simulated a 60\(\times \)42 grid network to have a setup that can be
representative of a current commercial one~\cite{bassa_analytical_2022}, even
though the developed routing module supports arbitrary satellite shells. In the
first experiment, we just wanted to test the behavior of the caches and the
routing algorithm. To this end, we calculated the location of the five caching
nodes in the top-right quadrant \( \{(11, 0), (18, 0), (24, 0), (0, 8), (0, 15)
\} \).\footnote{Note that the caches are still intercalated. Cache \((0,24)\)
    cannot be in the vertical axe, because its maximum size is \((0,21)\).} Then we
enabled the NDN cache in all these nodes, and the corresponding ones in the
three remaining quadrants. Finally, with a producer located at the center---node
\((0,0)\)---, we requested the same content from 1000 random locations. Consumers
at each location ask just once for the content, so the same content is requested
1000 times. Figure~\ref{fig:density-plot} shows the number of transmissions
performed by each node during the whole experiment.
\begin{figure}
    \centering
    \includegraphics[width=\columnwidth]{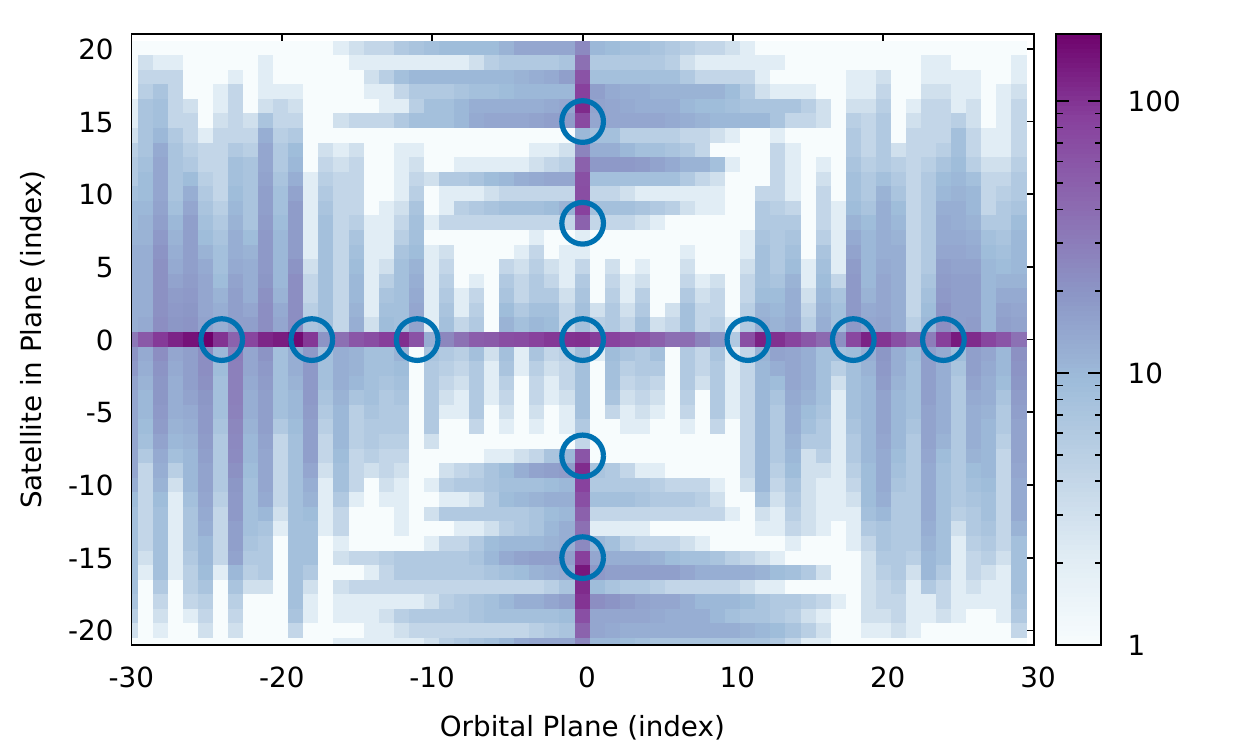}
    \caption{Number of transmission carried out by each node in a 60×42 grid
        network with 1000 randomly placed clients requesting the same data
        originally produced at node \((0,0)\). The circles show the locations of
        each caching node.}%
    \label{fig:density-plot}
\end{figure}
Nodes in both axes concentrate most of the transmissions, as the routing
algorithm drives the requests towards the axis with the closest allowed cache.
Also note that the mechanism is able to alleviate the load near the producer
since the number of transmissions carried out by nodes close to the producer is,
in fact, similar to those near caching nodes. Finally, one can discern in the
figure the areas covered by each caching node, as nodes in the boundaries
perform few transmissions as no traffic is directed to those regions, and
intensity grows higher the closer to a caching node.

The next experiment tests the accuracy of the theoretical results. In the same
network scenario as before, we measured the average transmission path length
until the data is obtained for different amount of caching nodes and a growing
number of randomly placed clients. The experiment has been repeated 25 times,
varying the location of each client.
\begin{figure}
    \centering
    \includegraphics[width=\columnwidth]{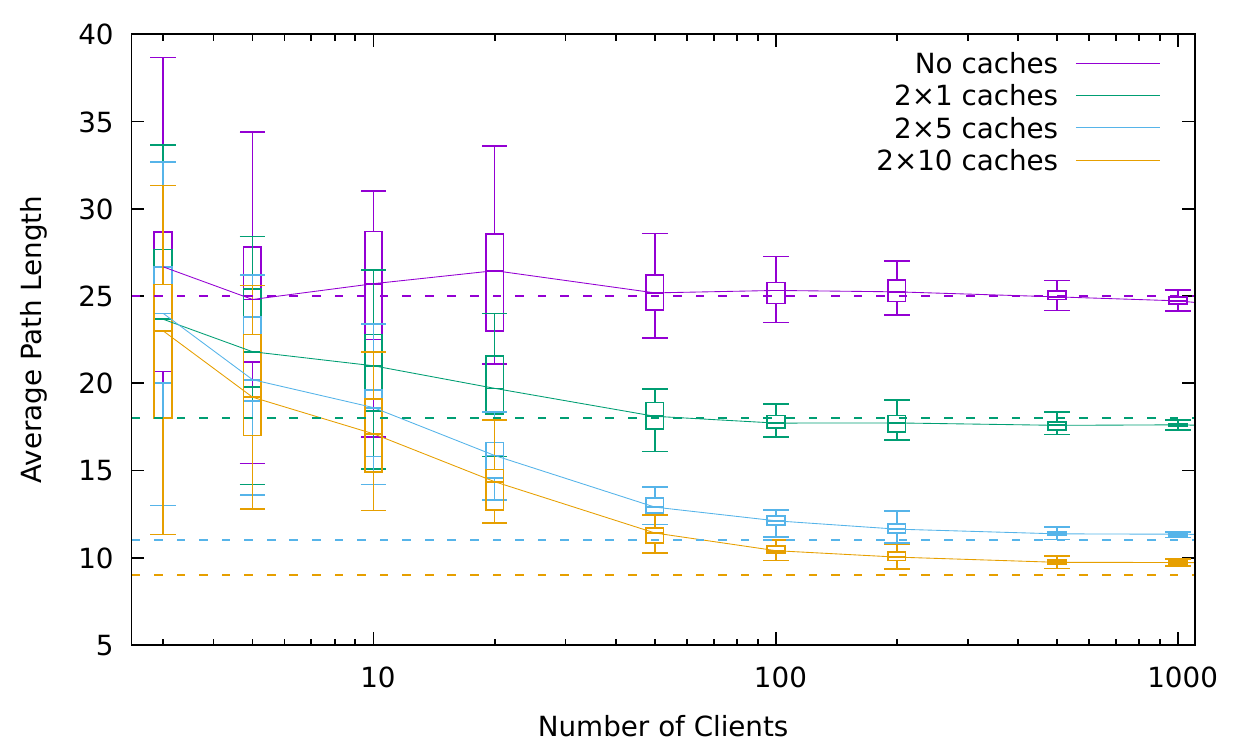}
    \caption{Evolution of the average path length vs.\ the number of clients.
        Hashed lines represent the theoretical values.}%
    \label{fig:path-vs-clients}
\end{figure}
Figure~\ref{fig:path-vs-clients} shows the results obtained when there are no
caching nodes, and for one, five and ten caching nodes in the positive parts of
the axes (and the corresponding set of caching nodes in the negative parts, thus
the \(2\times n\) notation). The theoretical values were calculated according
to~\eqref{eq:avdistance-region-axes-total} for the cache locations resulting
from Algorithm~\ref{alg:cache-placement-axes}. When the number of clients is
very small, there is a great variability in the results, as clients can be at very
different distances from the caches or the producer. As the number of clients
increases, and they get more evenly placed in the network, the variability
diminishes, and we can observe that the resulting average path length converges
to the theoretical value.

The simulation length also plays an important role in the results. To show this,
we have repeated the previous experiment but, this time, modifying the
simulation length for a constant value of 1000 randomly placed clients. The
request from each client happens at a random instant during the whole simulation. As
before, each simulation was repeated 25 times, varying the location of the
clients.
\begin{figure}
    \centering
    \includegraphics[width=\columnwidth]{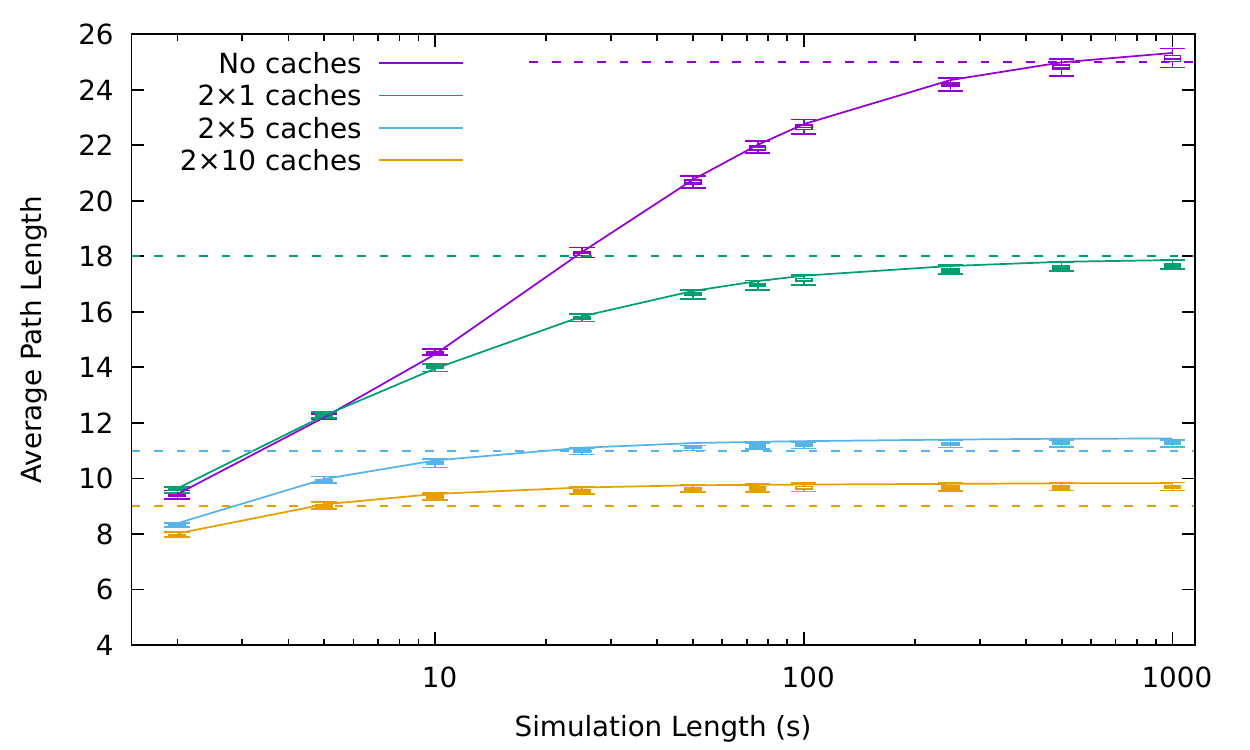}
    \caption{Evolution of the average path length vs.\ the simulation length. Hashed
        lines represent the theoretical values.}%
    \label{fig:path-vs-length}
\end{figure}
The results in Fig.~\ref{fig:path-vs-length} show that, when the simulation
duration is short, the results are much better than those predicted
by~\eqref{eq:avdistance-region-axes-total}. The reason for this is that, with
short simulation lengths, all the requests from the different clients happen in
a very short interval, so a second request from a different client can reach a
node that has still pending a previous request from another client, even if it is
not a caching node. As NDN coalesces requests for the same Interest, the effect
is similar as caching. When the request is finally satisfied by a downstream
node, all the pending requests will get a copy of the requested Data. As the
simulation duration increases, the number of simultaneous requests decreases and
the only effective data saving measure is the caching mechanism.\footnote{The
    longest simulation lengths are only included to show the asymptotic behavior of
    the algorithm. In an actual LEO constellation, the satellite serving a given
    producer changes every few minutes.}

Finally, we would like to pay attention to the different compromises between the
two considered alternatives for cache placement: regular cache placement and
in-axes cache placement. To this end we have calculated the resulting average
path length and the total number of caching nodes needed for both alternatives
for a simple 24\(\times \)24 network. We can observe in
Fig.~\ref{fig:caches-comparison-24x24} the average distance from a node to the
nearest cache for both the regular and the in-axes cache placement strategies.
It is evident that the greatest reduction in average distances happens for just
a few caching nodes.
\begin{figure}
    \centering
    \includegraphics[width=\columnwidth]{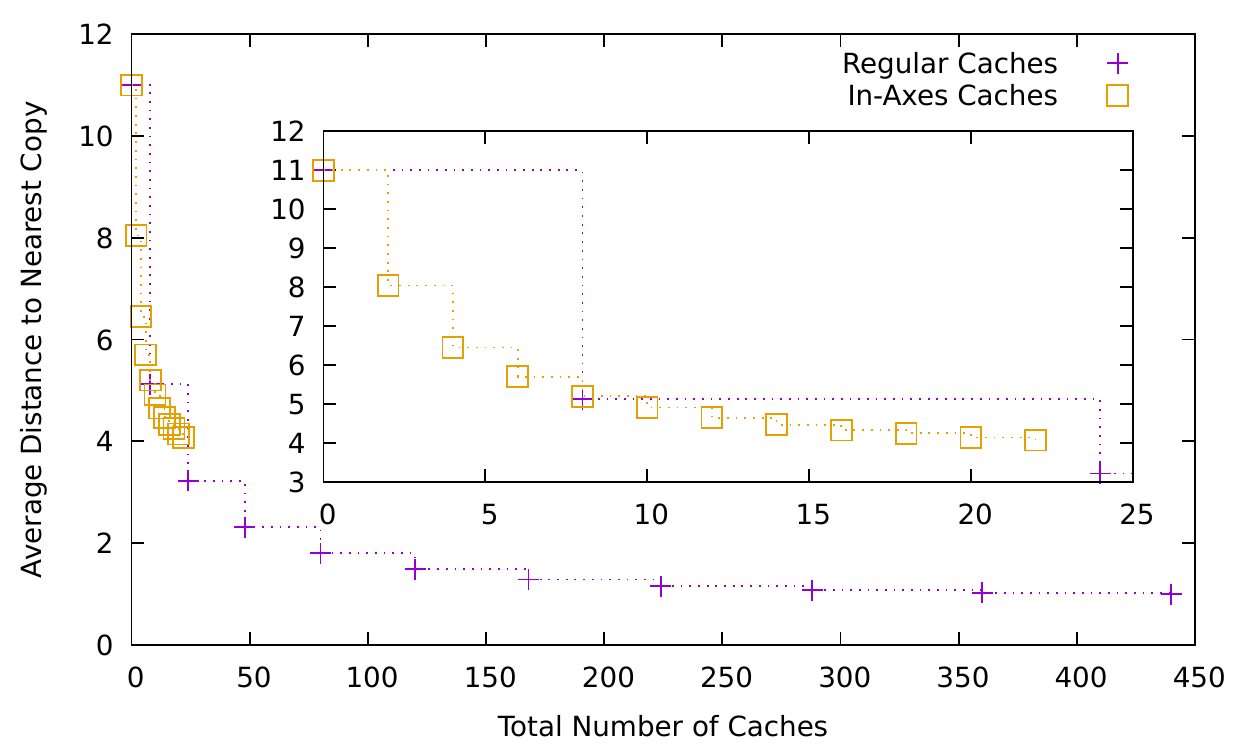}
    \caption{Comparison of cache efficacy between regular and in-axes cache
        placement strategies for a 24\(\times \)24 satellite constellation.}%
    \label{fig:caches-comparison-24x24}
\end{figure}
For instance, there is a \SI{45}{\percent} reduction for just 4 caching nodes
and a \SI{64}{\percent} reduction for 16 nodes. From there, as the number of caching
nodes increases, the improvement is marginal. For a small number of caching
nodes, both strategies produce very similar results, although only the regular
cache placement strategy is able to use a very high number of caching nodes.
However, when the number of caching nodes is small, the in-axes strategy
provides much more flexibility. Recall that, as shown
in~\eqref{eq:total-caches-sq}, the regular cache placement strategy places
stringent conditions on the possible total number of caching nodes.

\section{Conclusions}%
\label{sec:conclusions}

The recently deployed massive LEO satellite constellations serving as communication
backends for packet switched networks are an opportunity to explore new network
architectures that can be better suited to the job than the ubiquitous TCP/IP
one. The NDN architecture, for instance, with its inbuilt network caching capabilities,
may be used instead of custom CDN solutions to alleviate scarce network capacity and,
at the same time, reduce content delivery delay. This paper explores the issue
of cache placement considering the very regular structure of a satellite
constellation.

We have made a proposal able to simultaneously spread traffic through the
network, maximizing resource utilization and global capacity, and to concentrate
related traffic---traffic to a common producer---on a handful of network paths. Thus, by
placing caching nodes in these network paths, we can obtain high cache hit rates. In
our proposal, the decision about whether to cache a piece of data depends on the
relative locations of the forwarding node and the producer. As different nodes
are responsible for caching different pieces of content, the memory requirements
for the caching memory of the routers are equalized.

We have compared two different strategies for establishing the location of caching nodes
relative to the producers: an \emph{in-axes} alternative that places the caching
nodes either in the same orbital plane or in different orbital planes, but
similar latitude; and a second alternative that divides the constellation into
regular regions, each one served by a single caching node. For the in-axes cache
placement strategy, we have presented a linear-cost algorithm to obtain the optimal
location of a given number of caching nodes that minimizes path lengths. Experimental
results show that most of the performance is gained with just a few caching nodes per
piece of content and that, for a small number of caching nodes, both alternatives
produce similar results, although the in-axes approach is more flexible.

Future work includes exploring the possibility to apply these cache placement
algorithms in other regular network structures, like, for instance, the power
grid.

\section*{Acknowledgements}%
\label{sec:acks}

We wish to sincerely thank Margarita González-Romero for her insights and
suggestions for dealing with the proof of Lemma~\ref{lm:proof-optimal}.

\appendix{About the Interleaved Cache Placement}%
\label{sec:appendix}

\begin{lemma}%
    \label{lm:proof-intercalated}
    The optimum way to place the caches in the axes is in an interleaved manner,
    that is, \(h_i < v_i < h_{i+1}\).
\end{lemma}
\begin{proof}
    Consider a non-interleaved optimum solution to the problem \(\mathcal{H} =
    \{h_1, \dots, h_i, v_j, h_k, \dots, h_n\} \) and \( \mathcal{V} = \{v_1,
    \dots, v_i, h_j, v_k, \dots, v_n\} \), where \(h_i < v_i < h_{i+1}\).

    According to~\eqref{eq:cost-caches-axes-expanded}, the cost of this solution
    is
    \begin{equation}
        \begin{split}
            \C{\mathcal H  \cup \mathcal V}& = \C{(0, v_1, h_1)} + \C{(h_1, h_2, v_1)} +
            \C{(v_1, v_2,h_2)}\OBR
            \OCR{\quad} + \cdots + \C{(h_i, v_j,v_i)}+\C{(v_i, h_j,v_j)}\\
            & \quad + \C{(h_j, v_k,v_j)}+\C{(v_j, h_k,v_k)}\OBR
            \OCR{\quad} + \cdots + \C{(v_{n-1}, v_n, h_n)}.
        \end{split}
    \end{equation}
    If we swap now \(v_j\) and \(h_j\), so that \( \mathcal{H}' = \{h_1, \dots,
    h_i, h_j, h_k, \dots, h_n\} \) and \( \mathcal{V}' = \{v_1, \dots, v_i, v_j,
    v_k, \dots, v_n\} \), now the cost becomes
    \begin{equation}
        \begin{split}
            \C{\mathcal H'  \cup \mathcal V'}& = \C{(0, v_1, h_1)} + \C{(h_1, h_2, v_1)} +
            \C{(v_1, v_2,h_2)}\OBR
            \OCR{\quad} + \cdots + \C{(h_i, h_j,v_i)} + \C{(v_i, v_j, h_j)}\\
            & \quad + \C{(h_j, h_k, v_j)} + \C{(v_j, v_k, h_k)}\OBR
            \OCR{\quad} + \cdots + \C{(v_{n-1}, v_n, h_n)}.
        \end{split}
    \end{equation}

    The difference
    \begin{equation}
        \begin{split}
            \C{\mathcal H  \cup \mathcal V} - \C{\mathcal H'  \cup \mathcal V'}
            &= h_i h_j v_i - h_j v_i^2 - h_k v_j^2 \OBR
            \OCR{\quad} + {\left(h_j h_k - h_i v_i + v_i^2\right)} v_j \OBR
            \OCR{\quad}- {\left(h_j v_j - v_j^2\right)} v_k\\
            & = v_j{\left(v_i^2+v_j v_k - v_j h_k - h_i v_i\right)} \OBR
            \OCR{\quad} - h_j{\left(v_i^2+v_j v_k - v_j h_k - h_i v_i\right)}\OBR
            \OCR{} > 0
        \end{split}
    \end{equation}
    if
    \begin{equation}
        \label{eq:lema2-final-condition1}
        v_j > h_j,
    \end{equation}
    that is true because it is just a hypothesis of
    Lemma~\ref{lm:proof-intercalated}, and
    \begin{equation}
        \label{eq:lema2-final-condition2}
        \begin{split}
            v_i^2+v_j v_k - v_j h_k - h_i v_i\OCR{} >
            h_i v_i+v_j v_k - v_j h_k - h_i v_i\OBR
            \OCR{} > h_i v_i+v_j h_k - v_j h_k - h_i v_i\OBR
            \OCR{} > 0,
        \end{split}
    \end{equation}
    which is easy to check as both \(h_i < v_i\) and
    \(h_k < v_k\).
\end{proof}

\bibliographystyle{IEEEtaes}
\balance{}
\bibliography{icarus-axes}

\end{document}